\documentclass[journal]{IEEEtran} 
\ifCLASSINFOpdf
\else
\fi
\usepackage{setspace} 

\usepackage{geometry} 
\usepackage{cite}
\usepackage{ragged2e}
\usepackage{amssymb}
\usepackage{paralist}
\usepackage{amsmath}
\usepackage{amsthm}
\usepackage{graphicx}
\usepackage{epstopdf}
\usepackage{color}
\usepackage{extarrows}
\usepackage{multicol}
\usepackage{caption}
\usepackage{url}
\usepackage{bm}
\usepackage{footnote}
\usepackage{slashbox,multirow}
\usepackage{array}
\usepackage{diagbox}
\usepackage[lined,ruled,boxed,linesnumbered]{algorithm2e}
\usepackage{cases}

\usepackage[justification=centering]{subcaption}

\newtheorem{Prop}{Proposition}

\newtheorem{Pro}{Problem}
\newtheorem{Rmk}{Remark}
\allowdisplaybreaks
\PassOptionsToPackage{bookmarks=false}{hyperref}

\def\be{ \begin{eqnarray} }
\def\ee{ \end{eqnarray} }
\usepackage{soul}

\begin{document}
\title{Delay-Aware Flow Migration for Embedded Services in 5G Core Networks}

%

\author{\IEEEauthorblockN{Kaige Qu\IEEEauthorrefmark{1},
Weihua Zhuang\IEEEauthorrefmark{1}, 
Qiang Ye\IEEEauthorrefmark{1},
Xuemin (Sherman) Shen\IEEEauthorrefmark{1},
Xu Li\IEEEauthorrefmark{2}, 
and Jaya Rao\IEEEauthorrefmark{2}}

\IEEEauthorblockA{\IEEEauthorrefmark{1}Department of Electrical and Computer Engineering, University of Waterloo, Waterloo, Canada} 

\IEEEauthorblockA{\IEEEauthorrefmark{2}Huawei Technologies Canada Inc., Ottawa, ON, Canada
\\Email: \{k2qu, wzhuang, q6ye, sshen\}@uwaterloo.ca, \{Xu.LiCA, jaya.rao\}@huawei.com}}
\maketitle
\begin{abstract}
Service-oriented virtual network deployment is based on statistical resource demands of different services, while data traffic from each service fluctuates over time. In this paper, a delay-aware flow migration problem for embedded services is studied to meet end-to-end (E2E) delay requirement with time-varying traffic. A non-convex multi-objective mixed integer optimization problem is formulated, addressing the trade-off between maximum load balancing and minimum reconfiguration overhead due to flow migrations, under processing and transmission resource constraints and QoS requirement constraints. Since the original problem is non-solvable in optimization solvers due to unsupported types of quadratic constraints, it is transformed to a tractable mixed integer quadratically constrained programming (MIQCP) problem.  The optimality gap between the two problems is proved to be zero, so we can obtain the optimum of the original problem through solving the MIQCP problem with some post-processing. Numerical results are presented to demonstrate the aforementioned trade-off, as well as the benefit from flow migration in terms of E2E delay performance guarantee.
\end{abstract}


\IEEEpeerreviewmaketitle

\vspace{-1mm}
\section{Introduction}
\label{sec:Introduction}

The fifth generation (5G) networks will support multiple services with diverse service requirements. Customized network slicing is on a per-service basis over a common physical infrastructure to provide service level performance guarantee, which requires a flexible and programmable network architecture~\cite{Ye2018slicing}. Software defined networking (SDN) and network function virtualization (NFV) are two enabling technologies~\cite{herrera2016resource,nguyen2017sdn}. SDN decouples network control from data plane into a centralized control module, which facilitates global network management and enables network programmability. NFV separates network functions from dedicated hardware to software instances, referred to as virtual network functions (VNFs), operated in virtual machines (VMs) hosted in commodity servers or data centers (DCs), referred to as NFV nodes. With NFV, a network service is represented by a sequence of VNFs, constituting a service function chain (SFC). NFV is to enable elastic scaling of resources allocated to VNFs in a cost-effective manner, which achieves agile deployment and management of services.

With SDN and NFV, service requests are represented in the form of SFCs to satisfy certain expected demands according to service level agreement. Basically, VNFs are operated in VMs, referred to as VNF instances (VNFIs), placed on NFV nodes, and allocated certain processing resources. Traffic between consecutive VNFs is routed over physical paths, i.e., series of switches and links, referred to as tunnels, and allocated certain transmission resources. Both VNFIs and tunnels can be shared among multiple SFCs. This process is referred to as \emph{SFC embedding} in the virtual network planning phase~\cite{cohen2015near,li2017joint,Omar_toappear}. In the virtual network operation phase, traffic for each service arrives and fluctuates over time, possibly overloading some VNFIs and tunnels but underloading others from time to time. Imbalanced load can create bottlenecks on VNFIs and tunnels, leading to E2E delay violation for the affected services~\cite{Qiang_toappear}. To avoid E2E delay violation caused by the load-resource mismatch, we allow each SFC to traverse through alternative VNFIs and tunnels while obeying service chaining requirements, which is called \emph{flow migration}.

Flow migration for embedded SFCs is similar to dynamic SFC embedding, in terms of location change for VNFs and routing change for traffic between consecutive VNFs~\cite{liu2017dynamic,rankothge2017optimizing,eramo2017approach,guo2016dynamic}. However, they are different in timescale. In dynamic SFC embedding, the time-varying processing and transmission resource demands are known a priori, based on which VNFs are re-deployed on NFV nodes via VM migrations and traffic between consecutive VNFs is re-routed, generating new VNFIs and tunnels with different locations and resources~\cite{clark2005live}. The objective is to minimize total resource usage over both NFV nodes and links, considering reconfiguration overhead, i.e., a linear combination of the number of reconfigured NFV nodes and links~\cite{rankothge2017optimizing}. Flow migration operates in a smaller timescale with time-varying traffic, which is rerouted among the predetermined VNFIs and tunnels to satisfy E2E delay requirements. A balanced resource utilization makes the network more tolerant of future traffic changes, thus is considered as an objective for flow migration~\cite{guo2016dynamic}. Moreover, the transfer of VNF states should be considered, since simply rerouting the in-progress flows on a stateful VNF introduces state inconsistency, causing processing inaccuracy. For example, pattern matchings in intrusion detection systems (IDSs) should be transfered to the target VNFI on the rerouted path for accurate attack detection~\cite{gember2014opennf}. Some frameworks such as OpenNF are proposed to solve the state inconsistency problem, by moving not only the flow but also the VNF states~\cite{gember2014opennf}. Therefore, we also consider to minimize the number of migrations to reduce the reconfiguration overhead incurred by state transfers~\cite{nobach2017statelet}.

In this paper, we study a delay-aware flow migration problem for embedded SFCs in a processing-limited network with a focus on processing delay over VNFIs and ignoring transmission delay on tunnels, to achieve load balancing and to minimize reconfiguration overhead, under 1) service chaining requirements, 2) processing and transmission resource constraints, and 3) E2E delay requirements. The problem is formulated as a multi-objective mixed integer optimization problem. Our objective function to be minimized is a linear combination of 1) the maximum processing resource utilization ratio among VNFIs for load balancing, 2) the number of migrations for reduction of reconfiguration overhead due to state transfers, and 3) the number of extra flow rerouting tunnels for better utilization of available tunnels after SFC embedding. For delay awareness, we include the E2E delay requirements in the constraints with an assumption of prior knowledge of time-varying traffic rates. Due to several quadratic constraints, the optimization problem is non-solvable in optimization solvers such as Gurobi~\cite{gurobi}, so we transform it into a tractable MIQCP problem, and prove the optimality gap between the transformed problem and the original problem. An algorithm is proposed to obtain the optimal solution to the original problem given an optimal solution to the MIQCP problem.
\vspace{-1.5mm}

\section{System Model}
\label{sec:System Model} 

\subsection{Service Requests}
\label{sec:Service Requests}

\begin{figure}
\centering
\setlength{\belowcaptionskip}{-0.4cm}
{\includegraphics[width=1.0\linewidth]{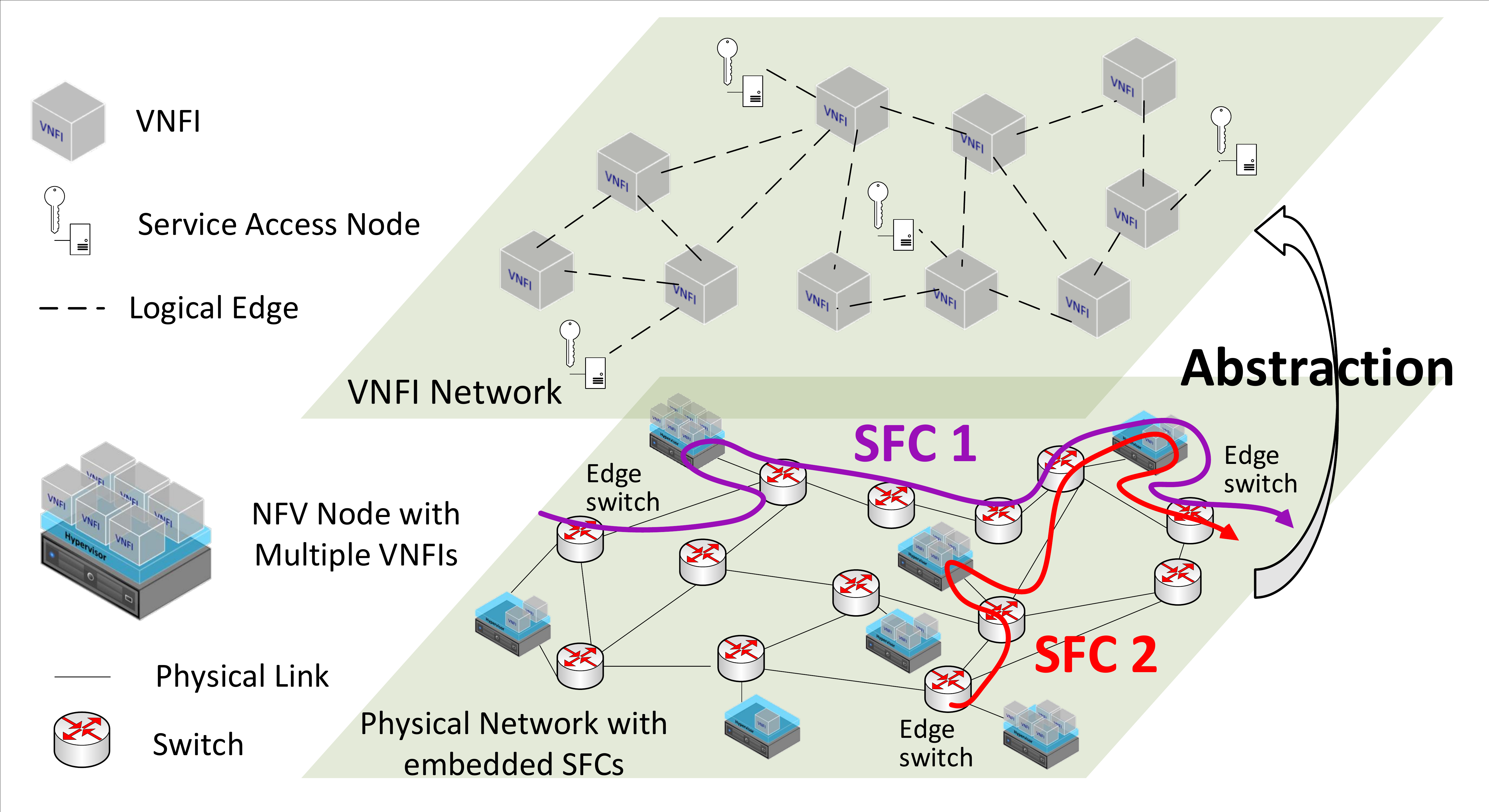}}
\caption{VNFI network abstraction.}\label{fig:abstraction}
\end{figure}

Let $\mathcal{R}$ denote the set of service requests. A service request, $r \in \mathcal{R}$, in the form of SFC, is represented as $\left\{a^{(r)}, V_{1}^{(r)},\cdots, V_{{H_r}}^{(r)}, b^{(r)}, D^{(r)}, \lambda^{(r)}(t)\right\}$. It traverses from source $a^{(r)}$, through $H_r$ VNFs in sequence, towards destination $b^{(r)}$, with $D^{(r)}$ and $\lambda^{(r)}(t)$ being respectively the E2E delay requirement and the traffic rate in packet/s during time interval $t$ in a time slotted system. Letting $\mathcal{H}_r =\{1,\cdots,H_r\}$, the $h$-th ($h \in \mathcal{H}_r$) VNF in SFC $r$ is denoted as $V_h^{(r)}$. Letting $\mathcal{F}$ be the set of VNF types supported by the network, e.g., firewall and IDS, we assume that there is at most one VNF of type $f \in \mathcal{F}$ in one SFC. A virtual link is a directed logical connection between two consecutive VNFs or between source/destination and its neighboring VNF. The $h$-th ($h\in \{0\}\cup\mathcal{H}_r$) virtual link in SFC $r$ is denoted as $L_h^{(r)}$, with $L_0^{(r)}$ connecting source $a^{(r)}$ and VNF $V_1^{(r)}$, $L_h^{(r)}$ ($h\in \mathcal{H}_r \backslash H_r$) connecting VNF $V_h^{(r)}$ and VNF $V_{h+1}^{(r)}$, and $L_{H_r}^{(r)}$ connecting VNF $V_{H_r}^{(r)}$ and destination $b^{(r)}$. Let $L$ denote the total number of virtual links in all SFCs. 
\vspace{-4mm}

\subsection{Physical Network with Embedded SFCs}
\label{sec:Physical Network}

The physical network represents a 5G core network, where the network nodes include switches and NFV nodes. Switches can forward traffic from incoming physical links to outgoing physical links. NFV nodes have both forwarding and processing capabilities. Each NFV node can support multiple VMs hosting different functions. The VNFs and virtual links belonging to different SFCs have been embedded to the physical network, based on long-term resource demands. Every VNF is embedded to a single NFV node, but every virtual link is allowed to be embedded to multiple physical paths. For a virtual link connecting two VNFs, its embedded physcial paths must be paths between the two embedded NFV nodes. 
\vspace{-3.5mm}

\subsection{Virtual Network Function Instance Network Abstraction}
\label{sec:Virtual Network Function Instance Network Abstraction}

Assume that all VNFs of type $f \in \mathcal{F}$ embedded to the same NFV node are \emph{mapped} to one VNFI of type $f \in \mathcal{F}$ on the NFV node. We abstract a VNFI network from the physical network with embedded SFCs, and represent it as a directed graph $G=\{\mathcal{I}\cup\mathcal{A},\mathcal{E}\}$, where $\mathcal{I}$ is a set of VNFIs, $\mathcal{A}$ is a set of service access nodes containing all sources and destinations of the embedded SFCs, and $\mathcal{E}$ is a set of logical edges, as illustrated in Fig.~\ref{fig:abstraction}. Let $\mathcal{I}_h^{(r)}$ be a subset of VNFIs of the same type as VNF $V_h^{(r)}$, and $V_i$ be a subset of VNFs of the same type as VNFI $i\in\mathcal{I}$. The processing resource capacity $C_i$ of VNFI $i \in \mathcal{I}$ is the aggregate maximum processing rate in cycle/s of all VNFs mapped to it. Since the processing resource demand is service type and VNF type dependent, we define the processing density of VNFI $i$ for SFC $r$ as $P_i^{(r)}$ in cycle/packet, which is the CPU resource demand in cycle/s on VNFI $i$ as per packet/s of processing rate of SFC $r$~\cite{shin2008dual}. Accordingly, with a CPU share of $\eta_i^{(r)}$ for SFC $r$ on VNFI $i$, the allocated processing rate in packet/s is $\mu_i^{(r)}=\frac{\eta_i^{(r)} C_i}{P_i^{(r)}}$.  

There is a directed logical edge (i.e., tunnel) between $i,j \in \mathcal{I}\cup\mathcal{A}$ only if there is at least one virtual link whose starting point and ending point are respectively mapped to $i$ and $j$. In this case, we say that at least one virtual link is \emph{mapped} to the logical edge from $i$ to $j$. For logical edge $e \in \mathcal{E}$, we define two sets of binary parameters $\{\beta_i^e\}$ and $\{\varphi_i^e\}$ to describe its location and direction, with $\beta_i^e=1$ if $i\in \mathcal{I}\cup\mathcal{A}$ is its starting point and $\varphi_i^e=1$ if $i\in \mathcal{I}\cup\mathcal{A}$ is its ending point.
\vspace{-3.5mm}

\subsection{Reconfiguration Overhead}
\label{sec:Reconfiguration Overhead}

\begin{figure} 
\centering
\setlength{\belowcaptionskip}{-0.45cm}
{\includegraphics[width=0.75\linewidth]{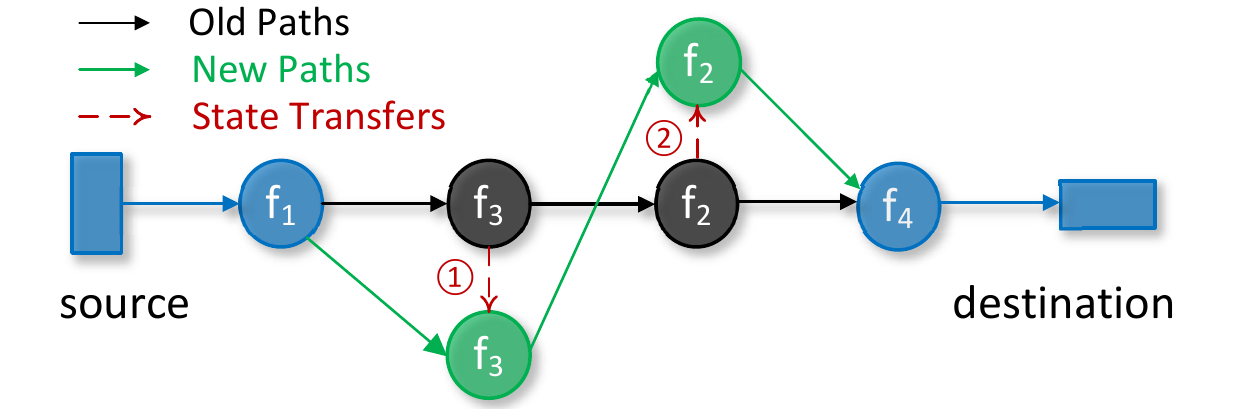}}
\caption{An illustration for flow migration and state transfer.}\label{fig:E-SFC-migration-unicast}
\end{figure}

When an SFC flow is migrated at a stateful VNF, the VNF is remapped to an alternative VNFI, with an associated state transfer, as illustrated in Fig.~\ref{fig:E-SFC-migration-unicast}. State transfers happen between two VNFIs of the same type. However, there is no direct logical edge between such two VNFIs. Therefore, extra logical edges should be created to support state transfers in date plane, incurring transmission resource overhead. In addition, some logical edges on the new paths after rerouting may not exist in the VNFI network, so extra logical edges for flow rerouting should also be created, incurring signaling overhead for the rerouted flows. Therefore, we describe the reconfiguration overhead in two parts: the number of extra logical edges for state transfers (equal to the number of migrations), and the number of extra logical edges for flow rerouting.
\vspace{-3mm}

\section{Problem Formulation}
\label{sec:Problem Formulation}

\begin{Pro}
Consider a VNFI network $G=\{\mathcal{I}\cup\mathcal{A},\mathcal{E}\}$. Packet processing time at a VNFI for a certain SFC is exponentially distributed. During a time interval, $t$, the traffic arrival to SFC $r\in\mathcal{R}$ is Poisson with mean rate $\lambda^{(r)}(t)$ packet/s. A delay-aware flow migration problem is to 1) find the remapping between VNFs and VNFIs in time interval $t$, on the basis of the mapping results in time interval $(t-1)$, and 2) reallocate processing resources on VNFIs, to satisfy the E2E delay requirements without violating the resource constraints. The objective is to achieve load balancing among all the VNFIs, with minimal reconfiguration overhead:
\end{Pro}

{\setlength\abovedisplayskip{-12pt} 
\setlength\belowdisplayskip{1pt}
\be
\begin{aligned}
    \min ~ O(t)=\hspace{0.3mm} \alpha_ 1 \hspace{0.3mm} \eta(t) + \alpha_2 \sum_{r \in \mathcal{R}} \sum_{h \in \mathcal{H}_r} \sum_{i,j \in \mathcal{I}_h^{(r)}} {m_{i\to j}^{rh}(t)} 
    \\
    + \hspace{0.8mm}\alpha_3 [L-\sum_{r \in \mathcal{R}} \sum_{h \in \{0\}\cup\mathcal{H}_r} \sum_{i,j \in \mathcal{I}\cup\mathcal{A}}  {y_{ij}^{rh}(t)} ]. 
    \label{eq-Obj}
\end{aligned}
\ee
}

In objective function~\eqref{eq-Obj}, there are several decision variables which are functions of $t$: 1) continuous function $\eta(t)\in[0,1]$ for resource utilization ratio of the busiest VNFI in the VNFI network during interval $t$; 2) binary function set $\boldsymbol{m}(t)=\{m_{i\to j}^{rh}(t)\}$, with $m_{i\to j}^{rh}(t)=1$ if the mapped VNFI for VNF $V_h^{(r)}$ changes from VNFI $i$ during interval $(t-1)$ to VNFI $j$ during interval $t$, and $m_{i\to j}^{rh}(t)=0$ otherwise; 3) binary function set $\boldsymbol{y}(t)=\{y_{ij}^{rh}(t)\}$, with $y_{ij}^{rh}(t)=1$ if virtual link $L_h^{(r)}$ is mapped to an existing logical edge between $i,j\in\mathcal{I}\cup\mathcal{A}$ during interval $t$, and $y_{ij}^{rh}(t)=0$ otherwise. Note that $\alpha_1,\alpha_2,\alpha_3$ are tunable weights to control the priority of the three components. In the right hand side of \eqref{eq-Obj}, the first term is the cost for imbalanced load among the VNFIs, the second term is the number of migrations, and the third term is the number of extra logical edges for flow rerouting.

Next, for constraints, first consider VNF to VNFI mapping. Define binary function set $\boldsymbol{x}(t) = \{x_i^{rh}(t)\}$, with $x_i^{rh}(t)=1$ if VNF $V_h^{(r)}$ is mapped to VNFI $i \in \mathcal{I}_h^{(r)}$ during interval $t$, and $x_i^{rh}(t)=0$ otherwise. Since VNF $V_h^{(r)}$ should be mapped to exactly one VNFI in $\mathcal{I}_h^{(r)}$, we have constraint 
\be
    \sum_{i \in \mathcal{I}_h^{(r)}} {x_i^{rh}(t)} = 1, \quad \forall (r,h) \in \mathcal{V}.   
    \label{eq-VNFI-mapping}
\ee
From interval $(t-1)$ to interval $t$, VNF to VNFI mapping changes from $\{x_i^{rh}(t-1)\}$ to $\{x_i^{rh}(t)\}$, where $\{x_i^{rh}(t-1)\}$ is a known set for interval $t$, so there is a relationship among $\{m_{i\to j}^{rh}(t)\}$ ($i\neq j$), $\{x_i^{rh}(t-1)\}$ and $\{x_{j}^{rh}(t)\}$, given by
\be 
\begin{aligned} 
    m_{i\to j}^{rh}(t) = x_i^{rh}(t-1) \hspace{0.3mm} x_j^{rh}(t), \quad\quad\quad\quad\quad\quad\quad\quad\quad
    \\
    \quad \forall (r,h) \in \mathcal{V}, ~ \forall i \in \mathcal{I}_h^{(r)}, ~ \forall j \in \mathcal{I}_h^{(r)}\backslash \{i\}.
    \label{eq-migration}
\end{aligned}
\ee
Also, we have 
\be
	m_{i\to i}^{rh}(t)=0, \quad \forall (r,h) \in \mathcal{V},  ~ \forall i \in \mathcal{I}_h^{(r)}.
	\label{eq-no-migration-ii}
\ee

Define nonnegative continuous function set $\boldsymbol{\mu}(t)=\{\mu_{i}^{rh}(t)\}$, with $\mu_{i}^{rh}(t)$ being the processing rate in packet/s allocated to VNF $V_h^{(r)}$ by VNFI $i \in \mathcal{I}_h^{(r)}$ during interval $t$. It should be lower bounded by $x_i^{rh}(t) \hspace{0.3mm} \lambda^{(r)}(t)$ due to the queue stability requirement and upper bounded by $x_i^{rh}(t) \hspace{0.3mm} \frac{C_i}{P_i^{(r)}}$ due to the limited processing capacity, given by 
\be
\begin{aligned} 
    x_i^{rh}(t) \lambda^{(r)}(t) \hspace{-0.6mm}\leq\hspace{-0.6mm} \mu_{i}^{rh}(t) \hspace{-0.6mm}\leq\hspace{-0.6mm} \frac{x_i^{rh}(t)C_i}{P_i^{(r)}}, ~\forall (r,h) \hspace{-0.6mm}\in\hspace{-0.6mm} \mathcal{V}, \hspace{0.5mm} \forall i \hspace{-0.6mm}\in\hspace{-0.6mm} \mathcal{I}_h^{(r)}.\hspace{-0.6mm}
    \label{eq-rate-bound}
\end{aligned}
\ee
Further, considering processing resource constraints with a maximum VNFI utilization ratio of $\eta(t)$, we have
{\setlength\belowdisplayskip{5pt}
\begin{subequations}
\be 
    & &\hspace{-1cm}
    \sum_{(r,h) \in \mathcal{V}_i} {{P_i^{(r)}\mu_{i}^{rh}(t)} \leq \eta(t){C_i}}, \quad \forall i \in \mathcal{I}
    \label{eq-polling}  
    \\
    & &\hspace{-1cm}
    0 \leq \eta(t) \leq 1.
    \label{eq-eta} 
\ee
\label{eq-resource}
\end{subequations}
}
Next, define positive continuous function set $\boldsymbol{d}(t)=\{d_{i}^{rh}(t)\}$, with $d_{i}^{rh}(t)$ denoting the average (dummy) processing delay on the processing queue associated with VNF $V_h^{(r)}$ at VNFI $i$. Since traffic to the queue is Poisson and packet processing time is exponential, the processing system is an M/M/1 queue. The delay $d_{i}^{rh}(t)$ is given by
\be 
\begin{aligned} 
    d_{i}^{rh}(t) \hspace{-0.6mm} = \hspace{-0.6mm} \frac{1}{\mu_{i}^{rh}(t) \hspace{-0.6mm} - \hspace{-0.6mm} x_i^{rh}(t) \lambda^{(r)}(t) \hspace{-0.6mm} + \hspace{-0.6mm} \epsilon}, \hspace{1mm} \forall (r,h) \hspace{-0.6mm} \in \hspace{-0.6mm} \mathcal{V}, \hspace{0.5mm} \forall i \hspace{-0.3mm} \in \hspace{-0.3mm} \mathcal{I}_h^{(r)}
    \label{eq-core-delay}
\end{aligned}
\ee
where $0<\epsilon\ll 1$ is a constant to avoid $d_{i}^{rh}(t)$ being undetermined, and $d_{i}^{rh}(t)$ is a dummy delay only if $x_i^{rh}(t)=0$. There is an extra bound constraint applied to $d_{i}^{rh}(t)$, explicitly showing its relationship with $x_i^{rh}(t)$:
\be 
\begin{aligned} 
    0 < d_{i}^{rh}(t) \leq x_i^{rh}(t) \hspace{0.3mm} D^{(r)} + \left[1-x_i^{rh}(t)\right] \hspace{0.3mm} \frac{1}{\epsilon}, 
    \\
    \forall (r,h) \in \mathcal{V}, ~ \forall i \in \mathcal{I}_h^{(r)}.
    \label{eq-d-bound}
\end{aligned}
\ee
The E2E processing delay of SFC $r \in \mathcal{R}$ should not exceed an upper bound $D^{(r)}$:
{\setlength\belowdisplayskip{-3pt}
\be
    \sum_{h \in \mathcal{H}_r} \sum_{i \in \mathcal{I}_h^{(r)}} {x_i^{rh}(t)\hspace{0.3mm}d_{i}^{rh}(t)} \leq D^{(r)},\quad \forall r \in \mathcal{R}.
    \label{eq-E2E-processing-delay} 
\ee
}

\vspace{-2mm}
According to the definition of $\{y_{ij}^{rh}(t)\}$, we have 
\begin{subequations} 
\be
    & &\hspace{-1cm}
    y_{ij}^{rh}(t) = \sum_{e \in \mathcal{E}} \beta_i^e \hspace{0.3mm} \varphi_j^e \hspace{0.3mm}  x_i^{rh}(t) \hspace{0.3mm}  x_j^{r(h+1)}(t), 
    \nonumber
    \\
    & &\hspace{-0.3cm}
    \forall r \in \mathcal{R}, \forall h \in \mathcal{H}_r\backslash \{H_r\},  \forall i \in \mathcal{I}_h^{(r)}, \forall j \in \mathcal{I}_{h+1}^{(r)}
    \label{eq-extra-edge1}
    \\
    & &\hspace{-1cm}
    y_{{a^{(r)}}j}^{r0}(t) = \sum_{e \in \mathcal{E}} \beta_{a^{(r)}}^e \hspace{0.3mm} \varphi_j^e \hspace{0.3mm} x_j^{r1}(t), \quad \forall r \in \mathcal{R}, \forall j \in \mathcal{I}_1^{(r)}
    \label{eq-extra-edge2}
    \\
    & &\hspace{-1cm}
    y_{i{b^{(r)}}}^{r H_r}(t) = \sum_{e \in \mathcal{E}} \beta_{i}^e \hspace{0.3mm} \varphi_{b^{(r)}}^e \hspace{0.3mm} x_i^{r{H_r}}(t), \quad \forall r \in \mathcal{R}, i \in \mathcal{I}_{H_r}^{(r)}
    \label{eq-extra-edge3}
    \\
    & &\hspace{-1cm}
    y_{ij}^{rh}(t) = 0, \quad otherwise.
    \label{eq-extra-edge4}
\ee
\label{eq-extra-edge}
\end{subequations}
\hspace{-1.5mm}Constraint \eqref{eq-extra-edge1} ensures that $y_{ij}^{rh}(t)$ equal $1$ only if VNF $V_h^{(r)}$ and VNF $V_{h+1}^{(r)}$ are respectively mapped to VNFI $i$ and VNFI $j$ between which a directed logical edge exists. Constraints \eqref{eq-extra-edge2} and \eqref{eq-extra-edge3} are applied to virtual links connecting either sources or destinations. There are exclusive conditions, e.g., $i \notin \mathcal{I}_h^{(r)}$ or $j \notin \mathcal{I}_{h+1}^{(r)}$, to which constraint \eqref{eq-extra-edge4} applies. 

We also consider transmission resource constraints on logical edges, i.e., the aggregated date rate on an existing logical edge should not exceed its transmission resource capacity:
\be
\sum_{r \in \mathcal{R}}\sum_{h \in \{0\}\cup\mathcal{H}_r}\sum_{i,j\in\mathcal{I}\cup\mathcal{A}} \hspace{-1.5mm} {\beta_i^e  \varphi_j^e  y_{ij}^{rh}(t) \lambda^{(r)}(t)\hspace{0.3mm}\sigma^{(r)}} \hspace{-0.6mm} \leq \hspace{-0.6mm} C_e, \hspace{1mm} \forall e \hspace{-0.3mm} \in \hspace{-0.3mm} \mathcal{E}
\label{eq-transmission}
\ee

\vspace{2mm}
Finally, we formulate the problem as
{\setlength\belowdisplayskip{-12pt}
\begin{subequations}
    \be
    & &\hspace{-0.9cm}
    \mathop {\min }\limits_{\scriptstyle \boldsymbol{x},\boldsymbol{m},\boldsymbol{\mu},\eta,\boldsymbol{y},\boldsymbol{d}\hfill} \; \hspace{0.35cm} {O(t) }
    \\
    & &\hspace{-0.4cm} \text{s.t. } \hspace{1.1cm} \eqref{eq-VNFI-mapping}-\eqref{eq-extra-edge}
    \\
    & &\hspace{1.2cm} \boldsymbol{x},\boldsymbol{m},\boldsymbol{y} \in \{0,1\}, \boldsymbol{d} >0.
    \ee
    \label{P1}
\end{subequations}
}

\begin{Rmk}
Problem \eqref{P1} is non-convex and cannot be solved in optimization solvers due to quadratic constraints \eqref{eq-core-delay}, \eqref{eq-E2E-processing-delay} and \eqref{eq-extra-edge1}. In Section~\ref{sec:MIQCP Reformulation}, we reformulate it as an MIQCP problem through constraint transformations~\cite{burer2012non}.
\end{Rmk}
\vspace{-4mm}

\section{MIQCP Reformulation}
\label{sec:MIQCP Reformulation}

By introducing auxiliary nonnegative continuous function set $\boldsymbol{\gamma}(t)=\{\gamma_{i}^{rh}(t)\}$, we get an equivalent linear form of constraint \eqref{eq-E2E-processing-delay} based on big-$\mathbb{M}$ method with $\mathbb{M}=\frac{1}{\epsilon}$ as
{\setlength\belowdisplayskip{-7pt} 
\begin{subequations} 
\be
    & &\hspace{-1.3cm}
    \sum_{h \in \mathcal{H}_r} \sum_{i \in \mathcal{I}_h^{(r)}} {\gamma_{i}^{rh}(t)} \leq D^{(r)},\quad \forall r \in \mathcal{R}
    \\
    & &\hspace{-1.3cm}
    d_{i}^{rh}(t) - \frac{1}{\epsilon}\left[1-x_i^{rh}(t)\right] \leq \gamma_{i}^{rh}(t) \leq d_{i}^{rh}(t), 
	\nonumber
    \\
    & &\hspace{2.1cm}
    \quad \forall (r,h) \in \mathcal{V}, \forall i \in \mathcal{I}_h^{(r)}
    \\
    & &\hspace{-1.3cm}
    0 \leq \gamma_{i}^{rh}(t) \leq \frac{1}{\epsilon} \hspace{0.3mm} x_i^{rh}(t),  \quad \forall (r,h) \in \mathcal{V}, \forall i \in \mathcal{I}_h^{(r)}.
\ee
\label{eq-E2E oricessing delay linear}
\end{subequations}
}

By introducing an auxiliary nonnegative binary decision function set $\boldsymbol{\xi}(t)=\{\xi_{ij}^{rh}(t)\}$, we get an equivalent linear form of constraint \eqref{eq-extra-edge1} for $\forall r \in \mathcal{R}, \forall h \in \mathcal{H}_r\backslash \{H_r\},  \forall i \in \mathcal{I}_h^{(r)}, \forall j \in \mathcal{I}_{h+1}^{(r)}$,
{\setlength\belowdisplayskip{-10pt} 
\begin{subequations}
    \be
    	& &\hspace{-1cm}
    	y_{ij}^{rh}(t) = \sum_{e \in \mathcal{E}} \beta_i^e \varphi_j^e \hspace{0.3mm} \xi_{ij}^{rh}(t)
    	\\
    	& &\hspace{-1cm}
   		\xi_{ij}^{rh}(t) \leq  x_i^{rh}(t)
        \\
    	& &\hspace{-1cm}
    	\xi_{ij}^{rh}(t) \leq x_j^{r(h+1)}(t)
        \\
    	& &\hspace{-1cm}
    	\xi_{ij}^{rh}(t) \geq  x_i^{rh}(t) + x_j^{r(h+1)}(t)-1.
    \ee
\label{eq-extra-edge1-linear}    
\end{subequations}
}

\begin{Prop}
If constraint \eqref{eq-core-delay} is replaced by 
\begin{subequations} 
    \be
    	& &\hspace{-1.5cm}
		\pi_{i}^{rh}(t) \hspace{-0.6mm} = \hspace{-0.6mm} \mu_{i}^{rh}(t) \hspace{-0.6mm} - \hspace{-0.6mm} x_i^{rh}(t) \lambda^{(r)}(t) \hspace{-0.6mm} + \hspace{-0.6mm} \epsilon, \hspace{1.5mm} \forall (r,\hspace{-0.3mm} h) \hspace{-0.6mm} \in \hspace{-0.6mm} \mathcal{V}, \forall i \hspace{-0.6mm} \in \hspace{-0.6mm} \mathcal{I}_h^{(r)}
		\label{eq-core-delay-linear-a} 
		\\
    	& &\hspace{-1.5cm}
   		\pi_{i}^{rh}(t) \geq \epsilon, \quad\quad\quad\quad \forall (r,h) \in \mathcal{V}, \forall i \in \mathcal{I}_h^{(r)}
   		\label{eq-core-delay-linear-m} 
    	\\
    	& &\hspace{-1.5cm}
   		d_{i}^{rh}(t) \hspace{0.3mm} \pi_{i}^{rh}(t) \geq c^2, \quad \forall (r,h) \in \mathcal{V}, \forall i \in \mathcal{I}_h^{(r)}
   		\label{eq-core-delay-linear-b} 
   		\\
    	& &\hspace{-1.5cm}
   		c = 1
   		\label{eq-core-delay-linear-c}
    \ee
\label{eq-core-delay-linear}    
\end{subequations}
\hspace{-2mm}where $\boldsymbol{\pi}(t)=\{\pi_{i}^{rh}(t)\}$ is an auxiliary continuous decision function set and $c$ is an auxiliary continuous decision variable, with constraint \eqref{eq-E2E-processing-delay} and \eqref{eq-extra-edge1} linearized, and the remaining other constraints, we obtain an MIQCP problem whose objective function and all constraints except the rotated quadratic cone constraint \eqref{eq-core-delay-linear-b} are linear. The optimality gap between the MIQCP problem and problem \eqref{P1} is zero, i.e., an optimal solution to problem \eqref{P1} is either a unique optimal solution or one of multiple optimal solutions to the MIQCP problem. Given an arbitrary MIQCP optimum (``$\star$''), an optimum of the original problem \eqref{P1} (``$\ast$'') can be obtained by Algorithm~\ref{alg:optimal solution to P1}.
\label{lem:Lemma01}
\end{Prop}

\begin{proof}
The difference between the MIQCP and the original problem lies in the ``$\geq$'' sign in constraint \eqref{eq-core-delay-linear-b}. If ``$\geq$'' is replaced by ``$=$'', the two problems are equivalent. Actually, an MIQCP optimum with inactive constraints in \eqref{eq-core-delay-linear-b} can always be mapped to another MIQCP optimum (original problem optimum) with active constraints in \eqref{eq-core-delay-linear-b}, without affecting other constraints and the objective value. The mapping algorithm is provided in Algorithm~\ref{alg:optimal solution to P1}. Detailed proof is omitted due to limited space and will be provided in an extended work.
\end{proof}

\begin{algorithm}[t]
\SetKwInOut{Input}{Input}{}
\SetKwInOut{Output}{Output}
\SetKwInput{Define}{Define}
\SetKwInput{Let}{Let}
\SetKwInput{Find}{Find}
\textbf{Input:} $\eta^{\star}$, $\boldsymbol{d}^{\star}$, $\boldsymbol{\mu}^{\star}$,$\boldsymbol{\tau}^{\star}$, $\boldsymbol{B}^{\star}$, $\boldsymbol{x}^{\star}$, $\boldsymbol{w}^{\star}$, $\boldsymbol{m}^{\star}$, $\boldsymbol{y}^{\star}$.\\	
\emph{Variable initialization} ($\ast = \star$).\\
\For{$r \in \mathcal{R}$, $\forall h \in \mathcal{H}_r$, $i \in \mathcal{I}_h^{(r)}$}{ 
\vspace{0.2em}
	\If{$d_{i}^{rh}(t)^{\star} \hspace{0.3mm} \pi_{i}^{rh}(t)^{\star} > (c^{\star})^2$ \emph{and} $x_i^{rh}(t)^{\star}==1$}{
		\lIf{ $\alpha_1>0$}{
			$d_{i}^{rh}(t)^{\ast}=\frac{1}{\mu_{i}^{rh}(t)^{\star}-x_i^{rh}(t)^{\star} \hspace{0.3mm} \lambda^{(r)}(t)+\epsilon}$
		}		
		\lIf{ $\alpha_1==0$}{
			$\mu_{i}^{rh}(t)^{\ast}=\frac{1}{d_{i}^{rh}(t)^{\star}}+x_i^{rh}(t)^{\star} \hspace{0.3mm} \lambda^{(r)}(t)-\epsilon$
		}  
	}
}
\lIf{$\alpha_1==0$}{calculate $\eta(t)^{\ast}$ ($\eta(t)^{\ast}\neq \eta(t)^{\star}$)}
\textbf{Output:} $\eta^{\ast}$, $\boldsymbol{d}^{\ast}$, $\boldsymbol{\mu}^{\ast}$,$\boldsymbol{\tau}^{\ast}$, $\boldsymbol{B}^{\ast}$, $\boldsymbol{x}^{\ast}$, $\boldsymbol{w}^{\ast}$, $\boldsymbol{m}^{\ast}$, $\boldsymbol{y}^{\ast}$\\	
\caption{Post-processing to MIQCP optimum}
\label{alg:optimal solution to P1}
\end{algorithm}

\vspace{-3mm}
\section{Numerical Results}

\begin{figure*}
    \centering
    \setlength{\belowcaptionskip}{-0.12cm}
    \begin{subfigure}[b]{0.31\textwidth}
        \includegraphics[width=\textwidth]{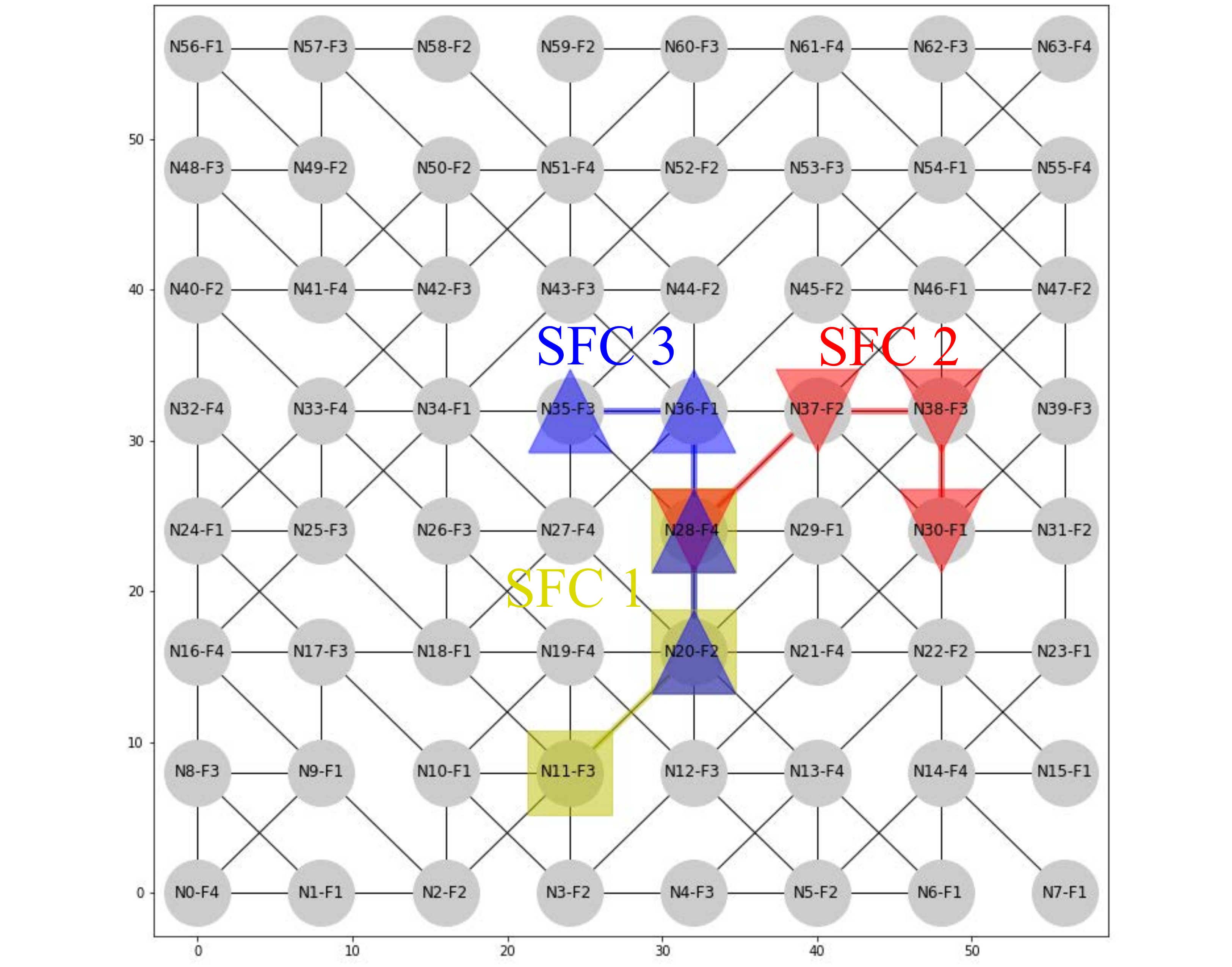}
        \caption{\footnotesize{The VNFI network and initial SFC mapping}}
    \end{subfigure}
    ~
    \begin{subfigure}[b]{0.31\textwidth}
        \includegraphics[width=\textwidth]{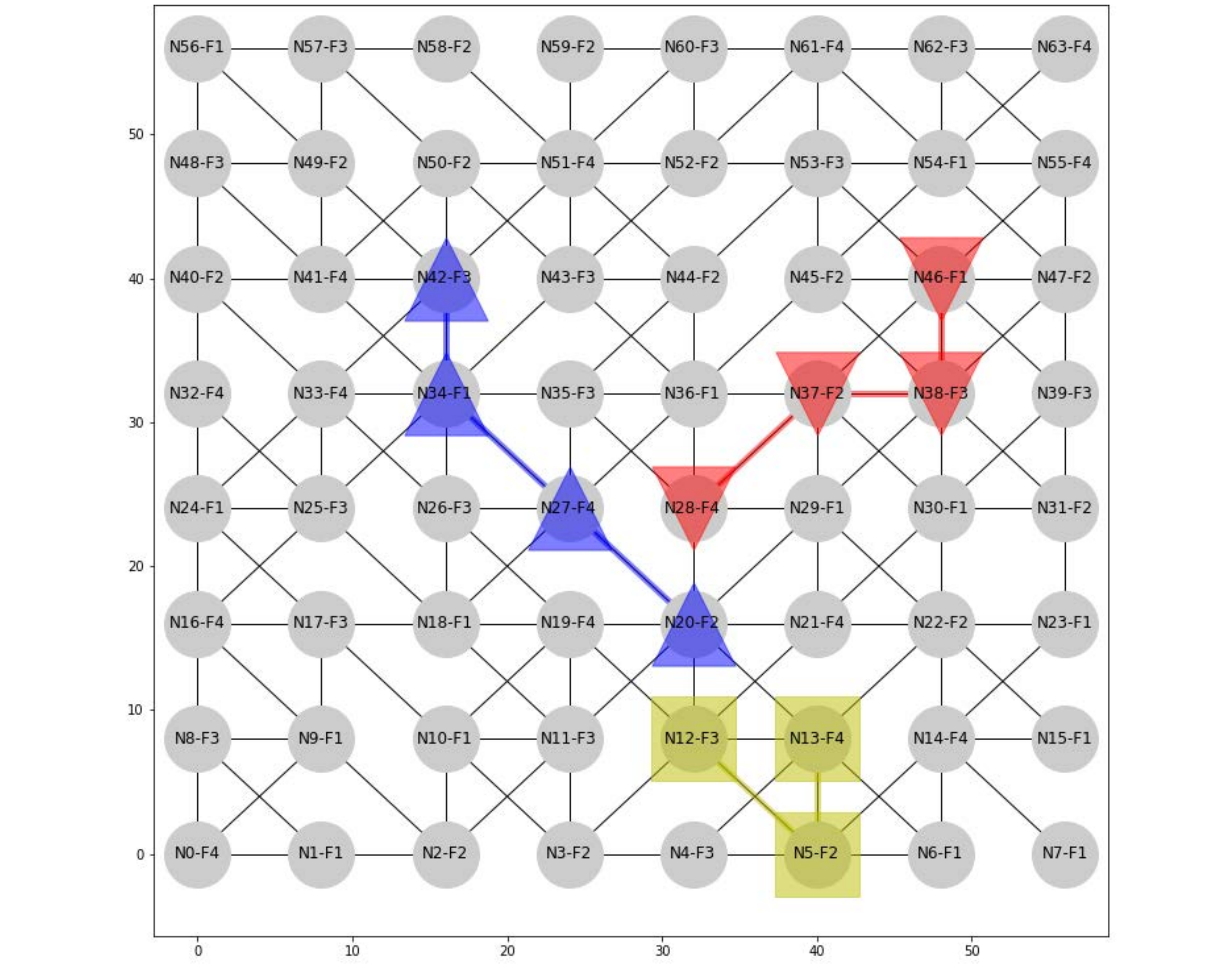}
        \caption{\footnotesize{Free migrations (LBFM)}}
    \end{subfigure}
    ~
    \begin{subfigure}[b]{0.31\textwidth}
        \includegraphics[width=\textwidth]{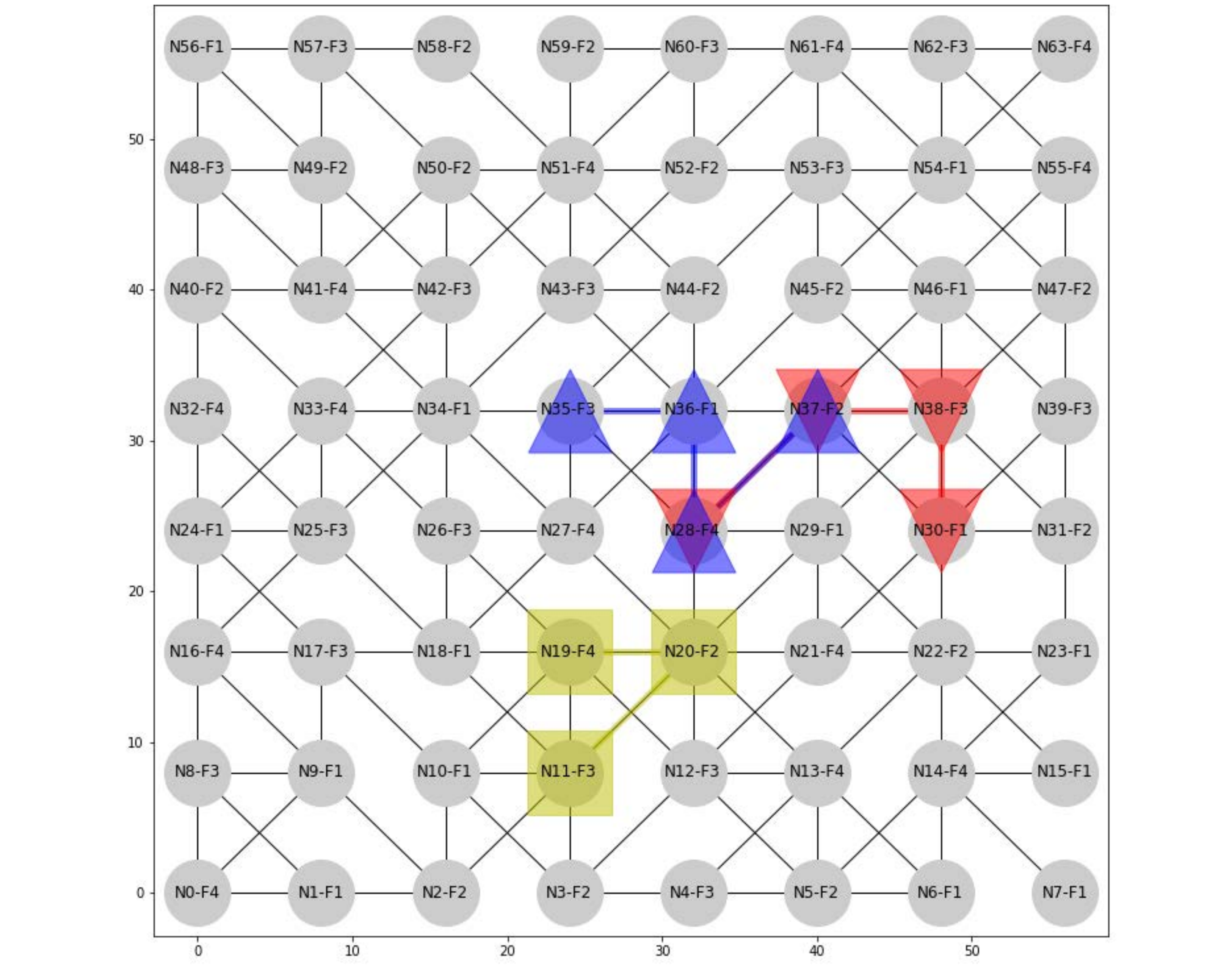}
        \caption{\footnotesize{Limited migrations (MOFM, HFM)}}
    \end{subfigure}
    \caption{VNFI network and SFC (re)mapping results at different flow migration strategies}\label{fig:VNFI network and mapping remapping results}
\end{figure*}

\begin{figure*}
    \centering
    \setlength{\belowcaptionskip}{-0.12cm}
    \begin{subfigure}[b]{0.324\textwidth}
        \includegraphics[width=\textwidth]{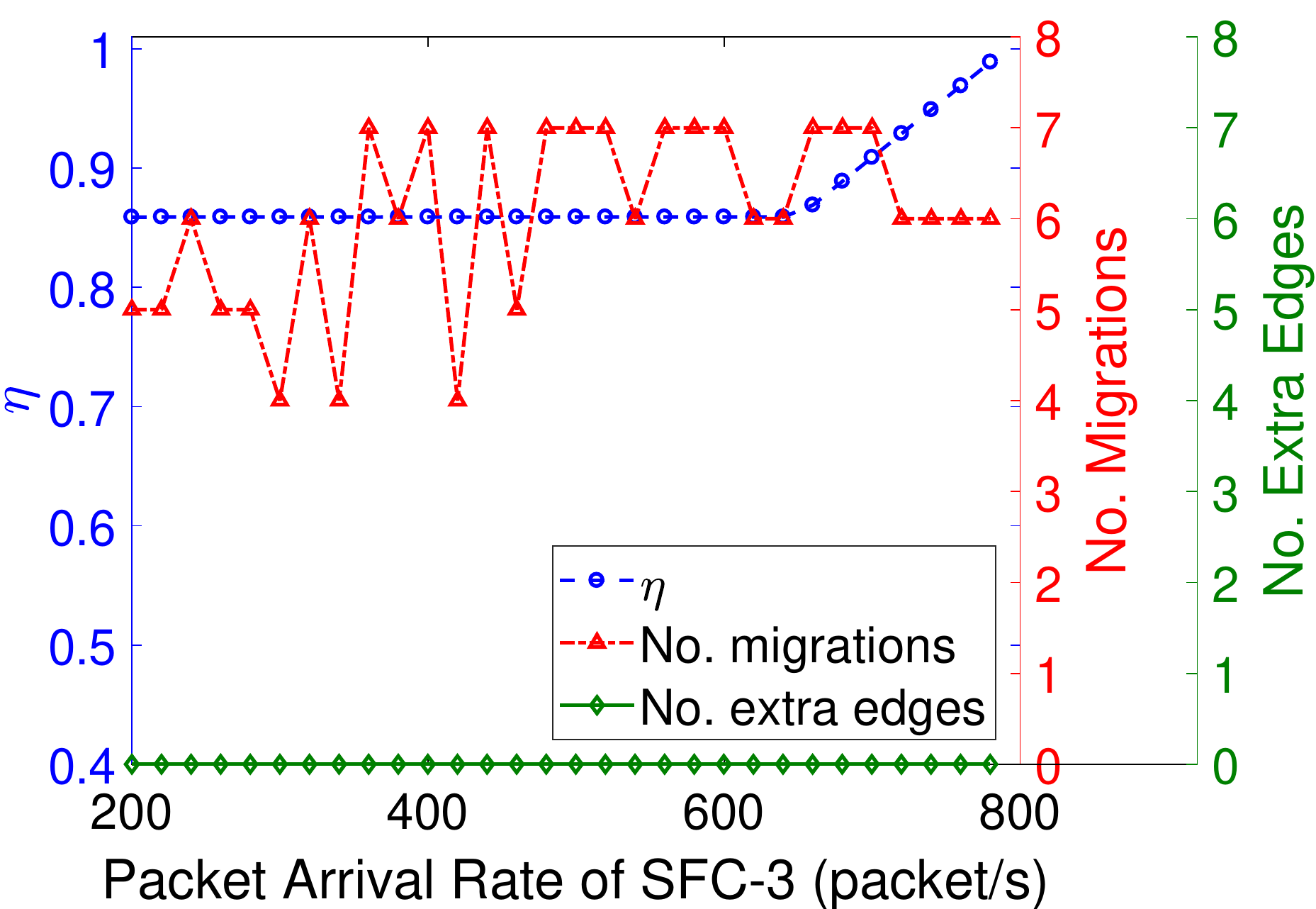}
        \caption{\footnotesize{LBFM strategy}}
    \end{subfigure}
    ~ 
    \hspace{-2mm}
    \begin{subfigure}[b]{0.324\textwidth}
        \includegraphics[width=\textwidth]{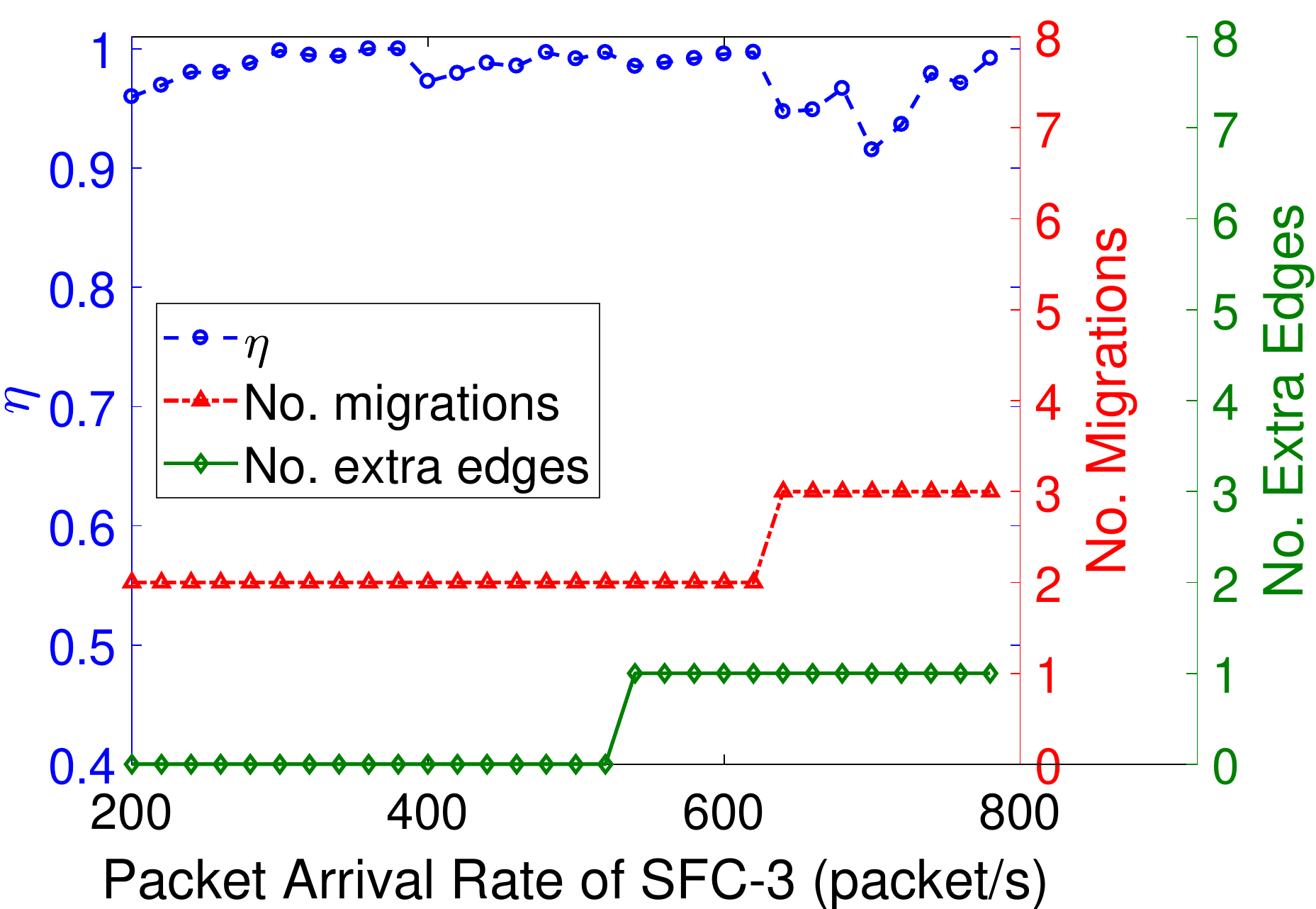}
        \caption{\footnotesize{MOFM strategy}}    
    \end{subfigure}
    ~ 
    \hspace{-2mm}
    \begin{subfigure}[b]{0.324\textwidth}
        \includegraphics[width=\textwidth]{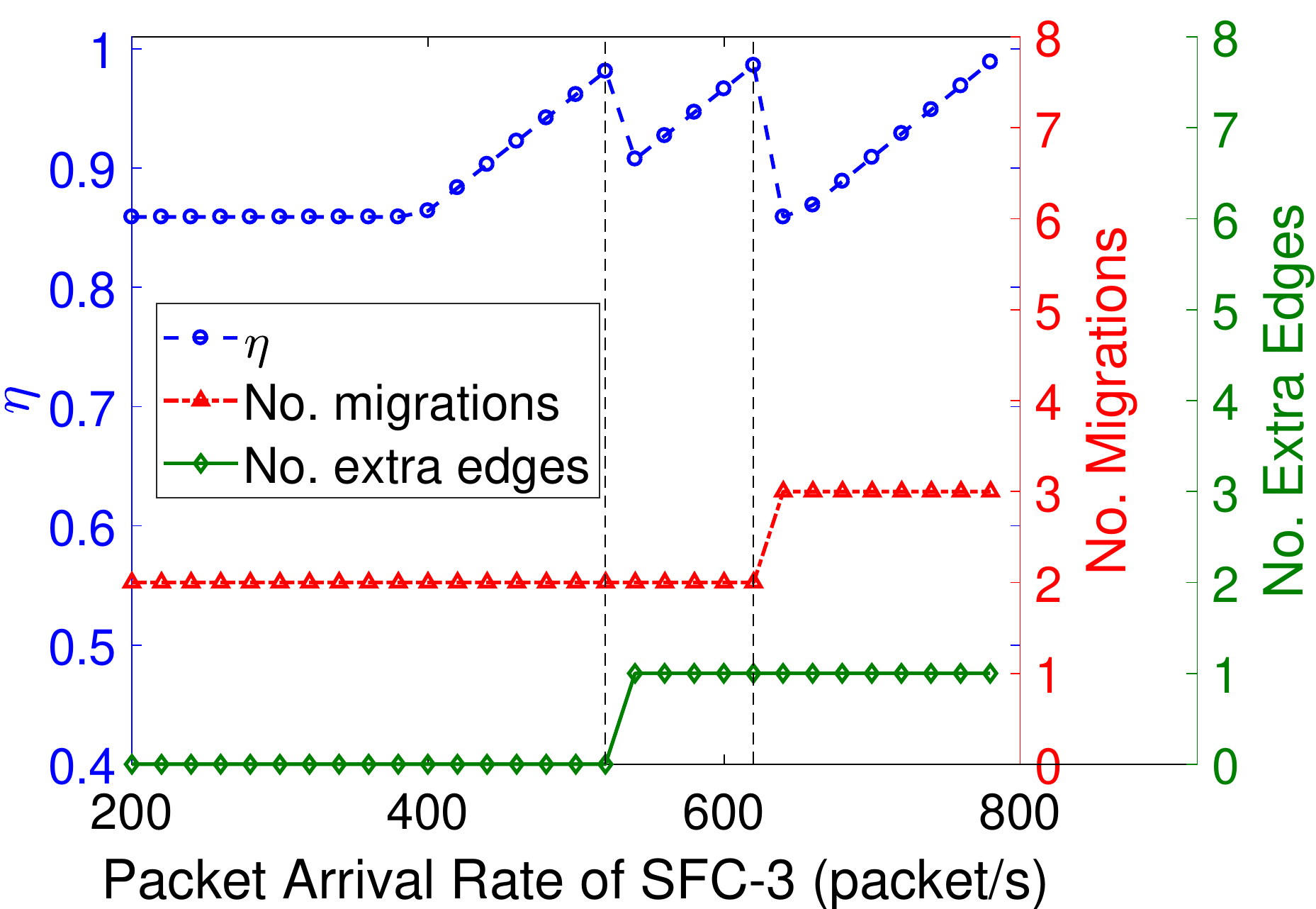}
        \caption{\footnotesize{HFM strategy}}
    \end{subfigure}
    \caption{Performance of three flow migration strategies with the increase of $\lambda^{(3)}(t)$ at $\lambda^{(1)}(t)=700$ packet/s}\label{fig:Performance of three flow migration strategies at large SFC1 traffic}
\end{figure*}

In this section, we present some numerical results for the flow migration problem solved using Algorithm~\ref{alg:optimal solution to P1}. We use a mesh topology in Fig.~\ref{fig:VNFI network and mapping remapping results} for the VNFI network, with 64 VNFIs. Four VNF types are considered, i.e., $\mathcal{F}=\{f_1,f_2,f_3,f_4\}$, with a random distribution in the network. The number of instances belonging to the four types are respectively 16, 16, 17 and 15. In terms of processing capacity, the VNFIs are homogeneous, with a processing rate of 1000 packet/s. Logical edges exist only between neighboring VNFIs with different types. There are 3 SFCs, including SFC 1 ($f_4 \to f_2 \to f_3$), SFC 2 ($f_1 \to f_3 \to f_2 \to f_4$) and SFC 3 ($f_3 \to f_1 \to f_4 \to f_2$), initially mapped to the VNFI network as shown in Fig.~\ref{fig:VNFI network and mapping remapping results}(a). Their sources and destinations are omitted for clarity. According to the initial mapping, VNFI ``N28-F4'' is shared by all three SFCs, and VNFI ``N20-F2'' is shared by SFC 1 and SFC 3. The E2E delay requirement of each SFC is set to 20 \emph{ms}.

\subsection{Load Balancing and Reconfiguration Overhead Trade-off}

We evaluate the performance of Algorithm~\ref{alg:optimal solution to P1} with varying traffic load under three sets of weights in~\eqref{eq-Obj}, to investigate the trade-off between load balancing and reconfiguration overhead. Performance metrics include $\eta(t)$, number of migrations, and number of extra logical edges for flow rerouting. Throughout the experiments, we fix $\alpha_3=0.2$, and explore three sets of \{$\alpha_1$, $\alpha_2$\}: \{0.8, 0\}, \{0, 0.8\} and \{0.4, 0.4\}, to focus on the trade-off between load balancing and number of migrations. For \{$\alpha_1$, $\alpha_2$\} = \{0.8, 0\}, the number of migrations is not optimized but load balancing is the focus, corresponding to a load balancing flow migration (LBFM) strategy. For \{$\alpha_1$, $\alpha_2$\} = \{0, 0.8\}, $\eta(t)$ is not optimized but migration reduction is emphasized, corresponding to a minimum overhead flow migration (MOFM) strategy. For \{$\alpha_1$, $\alpha_2$\} = \{0.4, 0.4\}, load balancing and migration reduction are both important, corresponding to a hybrid flow migration (HFM) strategy.

For traffic load, we fix $\lambda^{(2)}(t)=200$ packet/s, and vary both $\lambda^{(1)}(t)$ and $\lambda^{(3)}(t)$ from 200 packet/s to 780 packet/s. Beyond 780 packet/s, the problem becomes infeasible due to the processing resource constraints and the E2E processing delay constraints. Fig.~\ref{fig:VNFI network and mapping remapping results}(b)(c) illustrate the SFC remapping results of three flow migration strategies at $\{\lambda^{(1)}(t),\lambda^{(3)}(t)\}=\{700,700\}$ packet/s. For the LBFM strategy, SFCs completely separate from each other even at a relatively low traffic load to maximally balance traffic load in the VNFI network. However, only a limited number of migrations are observed in the MOFM and HFM strategies. Furthermore, we observe at most one extra logical edge for flow rerouting at different traffic loads throughout the experiment, verifying the effectiveness of the selected weight $\alpha_3=0.2$ for the third objective. Fig.~\ref{fig:Performance of three flow migration strategies at large SFC1 traffic} shows the performance of three flow migration strategies with variations of $\lambda^{(3)}(t)$ when $\lambda^{(1)}(t)$ is equal to 700 packet/s.

\begin{figure} 
\centering
\setlength{\belowcaptionskip}{-0.1cm}
    \begin{subfigure}[b]{0.35\textwidth}
        \includegraphics[width=\textwidth]{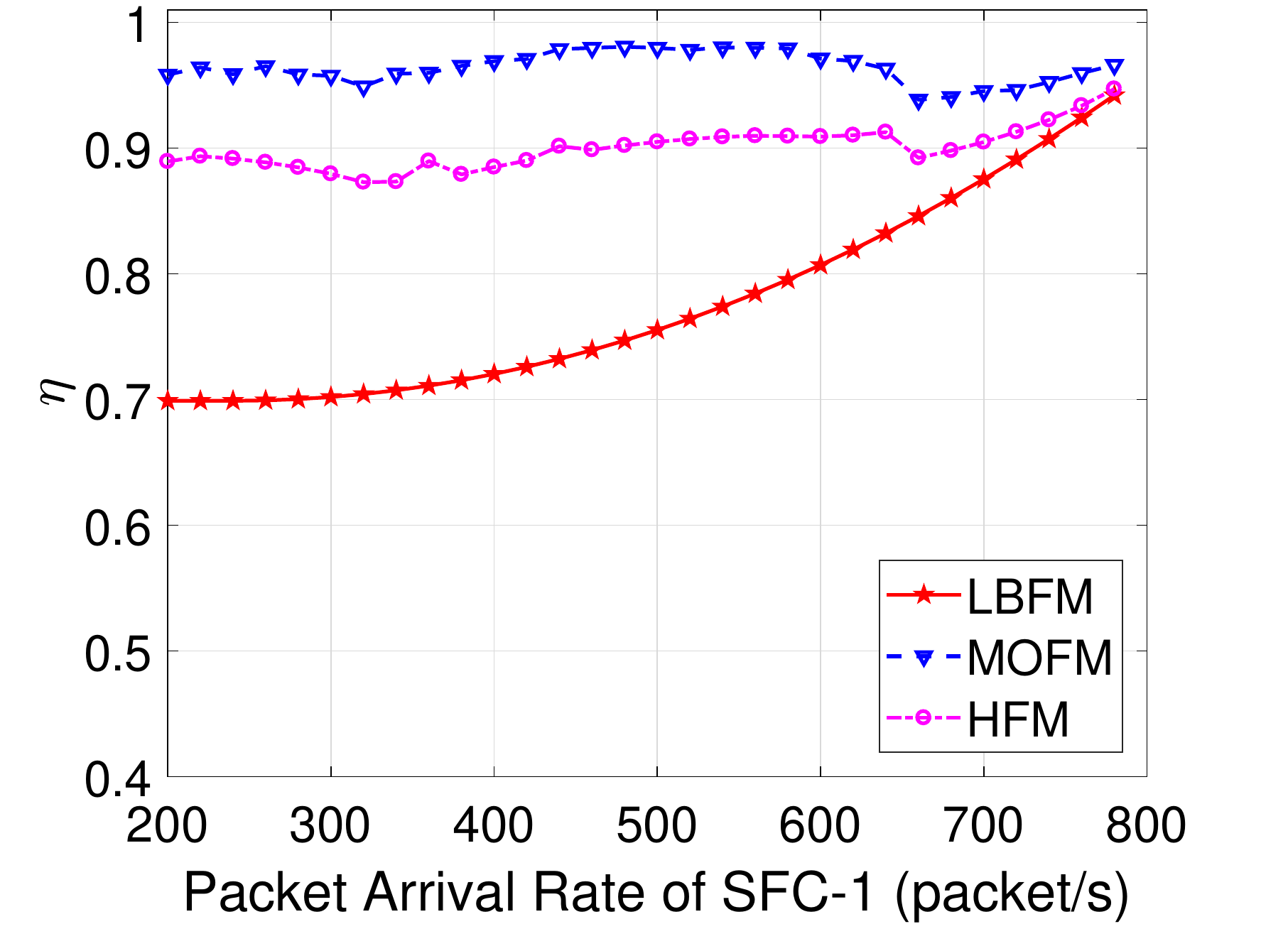}
    \end{subfigure}
    ~
    \begin{subfigure}[b]{0.35\textwidth}
        \includegraphics[width=\textwidth]{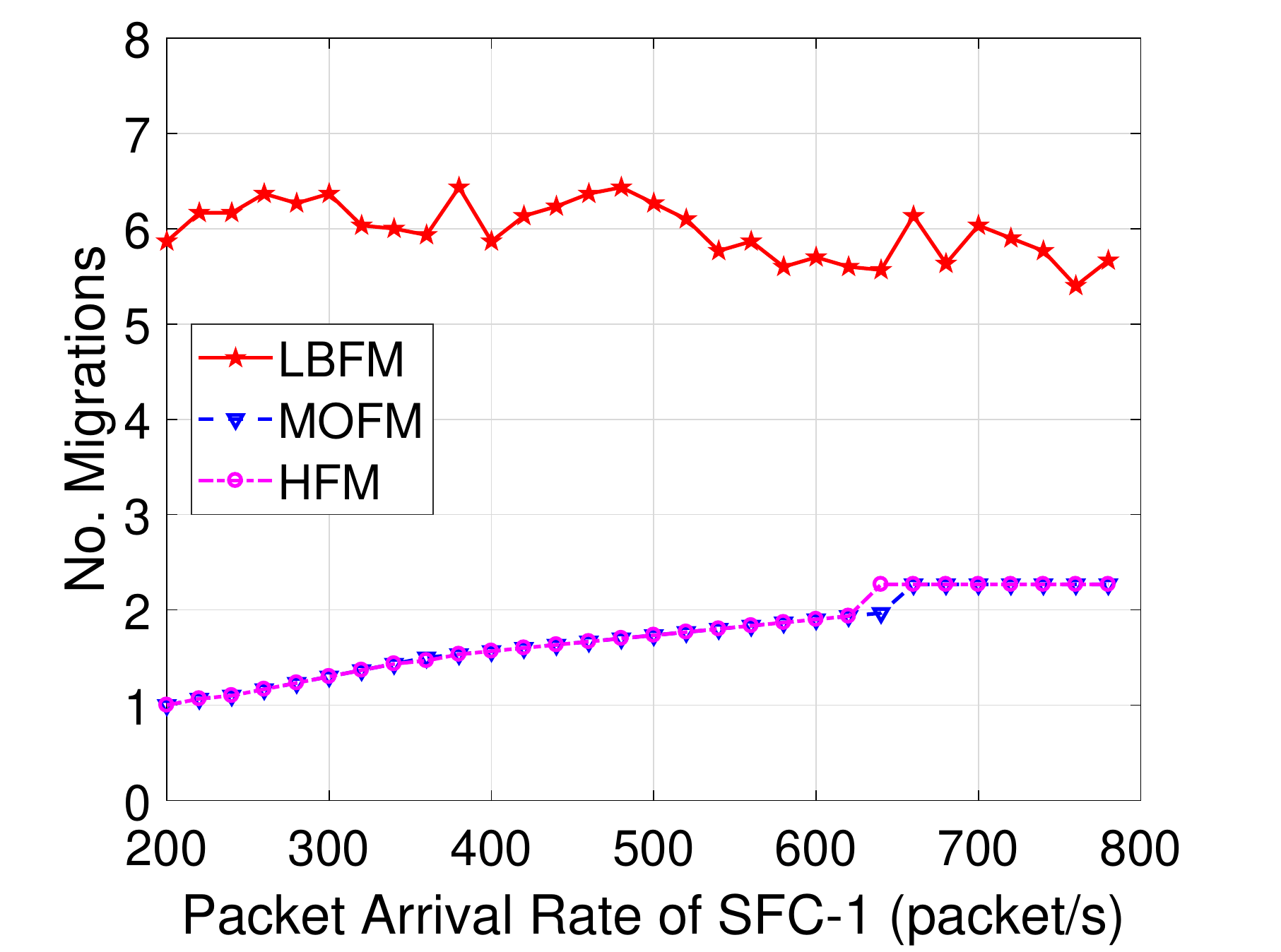}
    \end{subfigure}
\caption{Performance comparison of three flow migration strategies}\label{fig:Performance comparison of three flow migration strategies}
\end{figure} 

\subsubsection{LBFM strategy}

The maximum VNFI utilization ratio $\eta(t)$ is dominated by SFC 1 when $\lambda^{(3)}(t)$ is relatively small, showing a flat trend first, but turns to be dominated by SFC 3 with the increase of $\lambda^{(3)}(t)$. The number of migrations is high and random at different traffic loads. 

\subsubsection{MOFM strategy}

The number of migrations shows a step-wise increasing trend with the increase of $\lambda^{(3)}(t)$, taking a value from set $\{0,1,2,3\}$.  The maximum VNFI utilization $\eta(t)$ is close to 1 and shows a random relationship with the traffic load, since it is not optimized in the MOFM strategy. 

\subsubsection{HFM Strategy}

A trade-off between $\eta(t)$ and the other two objectives is observed. At a given $\lambda^{(1)}(t)$, $\eta(t)$ drops sharply when the number of migrations or the number of extra edges is increased by 1 with the increase of $\lambda^{(3)}(t)$. However, when the number of migrations and the number of extra edges stay stable, the curve of $\eta(t)$ shows either a linear increasing trend or a flat trend, with an increase of $\lambda^{(3)}(t)$. 

\subsubsection{Performance Comparison}

For the trade-off between load balancing and the number of migrations, Fig.~\ref{fig:Performance comparison of three flow migration strategies} shows the performance comparison of the LBFM, MOFM and HFM strategies as $\lambda^{(1)}(t)$ increases. For a given $\lambda^{(1)}(t)$ value, each performance metric is averaged across $\lambda^{(3)}(t) \in [200,780]$ packet/s. The number of migrations of the HFM strategy is comparable with that of the MOFM strategy, but it is much less than that of the LBFM strategy. At the same time, $\eta(t)$ of the HFM strategy is kept at a medium level as compared with the other two strategies.
\vspace{-3mm}

\subsection{End-to-End Delay Performance}

\begin{figure} 
\centering
    \begin{subfigure}[b]{0.355\textwidth}
        \includegraphics[width=\textwidth]{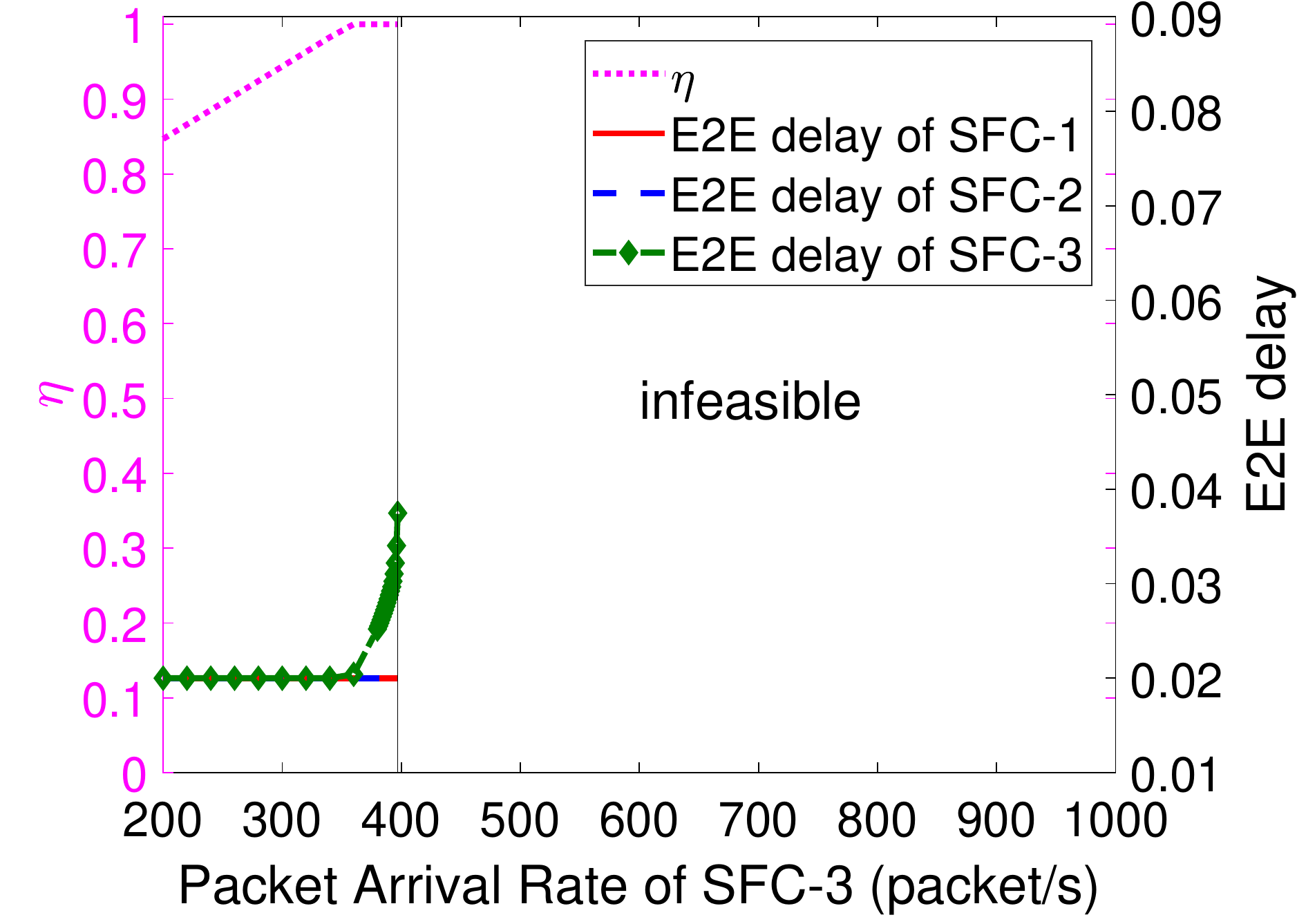}
        \caption{\footnotesize{Without flow migration}}
    \end{subfigure}
    ~
    \begin{subfigure}[b]{0.355\textwidth}
        \includegraphics[width=\textwidth]{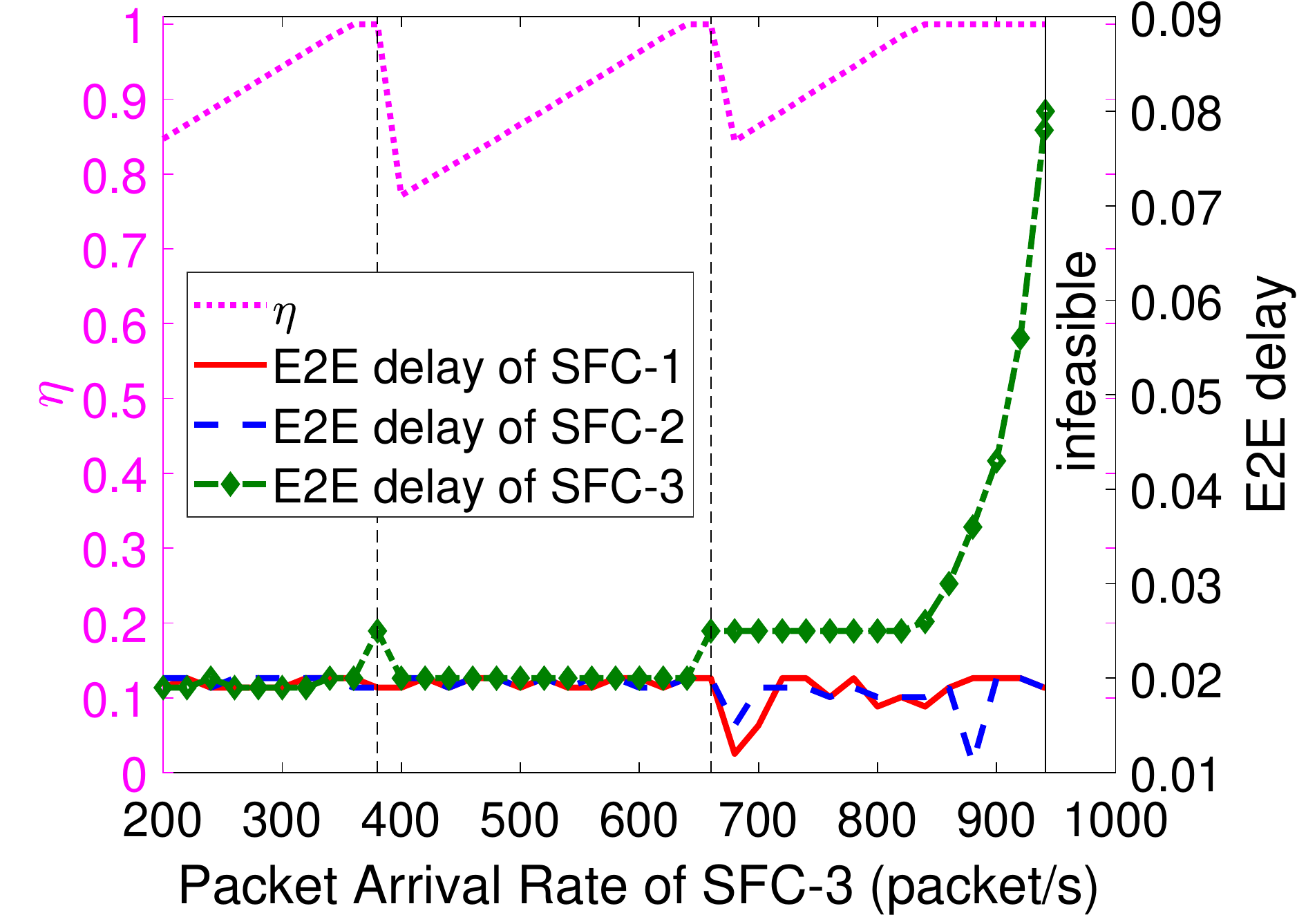}
        \caption{\footnotesize{With flow migration}}
    \end{subfigure}
\caption{End-to-end delay performance comparison}\label{fig:E2E delay performance}
\end{figure}

To evaluate benefits from flow migration on the E2E delay performance, we compare the delay with and without flow migration when the traffic rates increase. We set $\lambda^{(1)}(t)=\lambda^{(2)}(t)=200$ packet/s and increase $\lambda^{(3)}(t)$ from $200$ packet/s. Instead of the hard E2E delay constraints in \eqref{P1}, we allow some margins $\boldsymbol{\delta}=\{\delta^{(r)}\}$ to the E2E delay requirements $\sum_{h \in \mathcal{H}_r} \sum_{i \in \mathcal{I}_h^{(r)}} {x_i^{rh}(t)\hspace{0.3mm}d_{i}^{rh}(t)} \leq D^{(r)}+\delta^{(r)}$, and additionally minimize the aggregate delay margins in the objective function $O(t)+ \alpha_4 \sum_{r \in \mathcal{R}}{\delta^{(r)}}$. We confine $\delta^{(1)}=\delta^{(2)}=0$, and only optimize $\delta^{(3)}$. Fig.~\ref{fig:E2E delay performance} shows the E2E delay performance without and with flow migration. With the zigzag trend of $\eta(t)$, we can see at which traffic rates flow migrations happen. Benefiting from flow migrations, the feasible set of $\lambda^{(3)}(t)$ enlarges from [200, 397] packet/s to [200, 941] packet/s. 
\vspace{-2mm}

\section{Conclusion}
In this paper, we study a delay-aware flow migration problem for embedded SFCs in a common physical network in SDN/NFV enabled 5G systems. A mixed integer optimization problem is formulated based on a VNFI network abstraction, addressing the trade-off between load balancing and reconfiguration overhead, under service chaining requirements, processing resource constraints, and E2E delay constraints. The problem is non-convex and non-solvable in optimization solvers, so we reformulate a tractable MIQCP problem based on which the optimal solution to the original problem can be obtained. Numerical results show that the proposed model accommodates more traffic from the services which originally share some processing resources on their E2E paths, in comparison with a static SFC configuration without flow migrations. Moreover, a flow migration strategy with similar priority in load balancing and minimum number of migrations achieves medium load balancing, as compared with flow migration strategies with a bias on either goal. Nevertheless, it achieves approximately as good performance in terms of number of migrations as a flow migration strategy which is biased on migration reduction. This result indicates the benefit of joint consideration of the two goals.

\bibliographystyle{IEEEtran}
\bibliography{Reference_ICC19}

\end{document}